\documentclass[conference]{IEEEtran}
\IEEEoverridecommandlockouts
\usepackage{amssymb,amsfonts,bm}
\usepackage{amsthm}
\usepackage{balance}
\usepackage{cite}
\usepackage{amsmath,amssymb,amsfonts}
\usepackage{algorithmic}
\usepackage{graphicx}
\usepackage[shortlabels]{enumitem}
\usepackage{textcomp}
\usepackage{xcolor}
\usepackage{comment}
\usepackage{cite}
\newtheorem{theorem}{Theorem}[]

\newtheorem{lemma}[]{Lemma}
\newtheorem{claim}[]{Claim}
\newtheorem{definition}{Definition}[]
\newtheorem{proposition}{Proposition}

\newcounter{relctr} %% <- counter for relations
\everydisplay\expandafter{\the\everydisplay\setcounter{relctr}{0}} %% <- reset every eq
\newcommand\labelrel[2]{%
  \begingroup
    \refstepcounter{relctr}%
    \stackrel{\textnormal{(\alph{relctr})}}{\mathstrut{#1}}%
    \originallabel{#2}%
  \endgroup
}
\AtBeginDocument{\let\originallabel\label} %% <- store original definition

\def\BibTeX{{\rm B\kern-.05em{\sc i\kern-.025em b}\kern-.08em
    T\kern-.1667em\lower.7ex\hbox{E}\kern-.125emX}}
\begin{document}

\title{ Stochastic Bounded Confidence Opinion Dynamics: How Far Apart Do Opinions Drift?\\
%{%\footnotesize \textsuperscript{*}Note: Sub-titles are not captured in Xplore and should not be used}
}

\author{%
  \IEEEauthorblockN{Sushmitha Shree S, Kishore G V, Avhishek Chatterjee, Krishna Jagannathan}
  \IEEEauthorblockA{Department of Electrical Engineering\\         Indian Institute of Technology Madras\\ Chennai 600036, India\\
                    \{ee18d702, ee16b070\}@smail.iitm.ac.in, 
                    \{avhishek, krishnaj\}@ee.iitm.ac.in 
                    }
                    }

\begin{comment}
\author{\IEEEauthorblockN{Sushmitha Shree Sriram Kumar}
\IEEEauthorblockA{\textit{Department of Electrical Engineering} \\
\textit{IIT Madras}\\
Tamil Nadu, India \\
ee18d702@smail.iitm.ac.in}
\and
\IEEEauthorblockN{2\textsuperscript{nd} Given Name Surname}
\IEEEauthorblockA{\textit{dept. name of organization (of Aff.)} \\
\textit{name of organization (of Aff.)}\\
City, Country \\
email address or ORCID}
\and
\IEEEauthorblockN{3\textsuperscript{rd} Given Name Surname}
\IEEEauthorblockA{\textit{dept. name of organization (of Aff.)} \\
\textit{name of organization (of Aff.)}\\
City, Country \\
email address or ORCID}
\and
\IEEEauthorblockN{4\textsuperscript{th} Given Name Surname}
\IEEEauthorblockA{\textit{dept. name of organization (of Aff.)} \\
\textit{name of organization (of Aff.)}\\
City, Country \\
email address or ORCID}
\and
\IEEEauthorblockN{5\textsuperscript{th} Given Name Surname}
\IEEEauthorblockA{\textit{dept. name of organization (of Aff.)} \\
\textit{name of organization (of Aff.)}\\
City, Country \\
email address or ORCID}
\and
\IEEEauthorblockN{6\textsuperscript{th} Given Name Surname}
\IEEEauthorblockA{\textit{dept. name of organization (of Aff.)} \\
\textit{name of organization (of Aff.)}\\
City, Country \\
email address or ORCID}
}
\end{comment}

\maketitle
\begin{abstract}
In this era of fast and large-scale opinion formation, a mathematical understanding of opinion evolution, a.k.a. opinion dynamics, is especially important. Linear graph-based dynamics and bounded confidence dynamics are the two most popular models for opinion dynamics in social networks. Recently, stochastic bounded confidence opinion dynamics were proposed as a general framework that incorporates both these dynamics as special cases and also captures the inherent stochasticity and noise (errors) in real-life social exchanges. Although these dynamics are quite general and realistic, their analysis is particularly challenging compared to other opinion dynamics models. This is because these dynamics are nonlinear and stochastic, and belong to the class of Markov processes that have asymptotically zero drift and unbounded jumps. The asymptotic behavior of these dynamics was characterized in prior works. However, they do not shed light on their finite-time behavior, which is often of interest in practice. We take a stride in this direction by analyzing the finite time behavior of a two-agent system, which is fundamental to the understanding of multi-agent dynamics. In particular, we show that the opinion difference between the two agents is well concentrated around zero under the conditions that lead to asymptotic stability of the dynamics.

% In the broad literature of opinion dynamics, there are two basic models by which opinions of agents in a social system evolve - Linear dynamics wherein an agent is influenced by 
% its social neighbours, Bounded confidence opinion dynamics wherein an agent is influenced by those agents whose opinions are closely related (if the difference in the opinions is within a confidence value). Stochastic bounded confidence opinion dynamics is a novel variant of the latter dynamics that models probabilistic social exchanges between agents using influence functions. In this work, we consider a discrete-time stochastic bounded confidence opinion dynamics of two agents with real opinions. Often, the opinions of agents are not certainly known to each other in which case agents estimate each other's opinions and additionally, regardless of the strength of influence, agents are driven by self-beliefs. Using suitable noise processes, we model the difference in the opinions of agents as a Markov Process and infer that the expected translation distance in $n$ steps is $\mathcal{O}(n^{\frac{1}{2}-\beta})$ for some $\beta>0$. For stable dynamics, we characterise the evolution of opinion differences over a finite time window and provide approaches to deduce its concentration bound using Chernoff bound. 
%use abbreviations for stochastic bounded confidence opinion dynamics, if needed.
\end{abstract}


\section{Introduction}
\label{sec:intro}
Public opinion is the driving force of a society. The advent of social media platforms has revolutionized the speed and scale of opinion formation, resulting in significant effects on societies. Hence, modeling opinion formation, popularly known as opinion dynamics, and analyzing its behavior is a very important  problem.

The study of opinion dynamics has a long history and is well beyond the scope of this paper. In the mathematical and computational study of opinion dynamics, individuals or social entities, a.k.a. agents, are modeled to have real-valued opinions regarding a topic. A positive (negative) opinion represents a favorable (unfavorable) view of the topic and its magnitude represents the agent's conviction. Opinion dynamics models are discrete-time dynamical systems where opinions of the agents at the next time slot are updated according to a specified function of the current opinions.

Broadly, there have been two popular models of opinion dynamics: linear graph-based dynamics and bounded confidence dynamics. In the first model \cite{French1956, ABELSON1967, DeGroot1974, Friedkin1999, Fagnani2007, Acemoglu2013, Salehi2010, Yildiz2013}, opinion updates occur according to a linear combination of opinions of neighbors on a social graph. In the original bounded confidence dynamics \cite{Axelrod1997,Deffuant2000,Hegselmann2002}, an agent updates its opinion using the average of the opinion of all agents (including itself) whose opinions are within a specified distance from its own opinion. Thus, in short, the linear dynamics consider social graph-based opinion exchanges whereas the bounded confidence dynamics consider opinion-dependent opinion exchanges. 

There is a substantial amount of work extending both models to various scenarios and analyzing them either mathematically or numerically\cite{Noorazar2020}. Despite that, these models failed to capture two basic aspects of human interactions. First, when humans accept other opinions,  it is neither based on only acquaintance (as is the case with linear dynamics), nor it is based on a deterministic threshold (as in bounded confidence dynamics). Often, the consideration has inherent stochasticity, where the probability of accepting another agent's opinion decreases with increasing opinion differences. Second, the opinions of others are never known perfectly and can at best be estimated. In other words, there is always noise in opinion exchanges.

Stochastic bounded confidence (SBC) opinion dynamics proposed in \cite{Baccelli2015, Baccelli-Infocom, Baccelli2017} is a framework that addresses these two issues while merging the graph-based exchanges in linear dynamics with a stochastic generalization of the opinion dependent exchanges in bounded confidence dynamics. Unlike prior opinion dynamics models, where opinions eventually converge, the SBC dynamics capture real-life scenarios where opinions in a society may stay close in a probabilistic sense or may eventually diverge to opposite extremes. 

Linear dynamics and bounded confidence dynamics have been analyzed in detail in the literature. SBC dynamics is a stochastic generalization of both these dynamics on a graph and are characterized by high nonlinearity. Hence, their analysis is quite challenging. In \cite{Baccelli2015, Baccelli-Infocom, Baccelli2017}, specific conditions involving the social graph and the nature of the stochastic opinion-dependent exchanges were provided for limiting opinion differences to be finite. In multiple settings, tight converse results were also provided. Although these results display significant initial progress, they do not disclose anything about the evolution of opinions over a finite time window, which is often of interest in practice. In general, the nonlinear and stochastic nature of SBC dynamics makes the problem realistic and more challenging. 

In this paper, we take the first stride towards characterizing the opinion differences under SBC dynamics over a finite time window by studying a two-agent system. A two-agent system simplifies the SBC dynamics and allows to focus on the issues of nonlinearity and stochasticity by separating them from the complexity imposed by graph structures. Hence, understanding two-agent dynamics is necessary for studying the general SBC dynamics as it can offer useful insights. A two-agent system can also be viewed as  opinion dynamics between opposing groups of people in a society. Thus, in addition to being the first step towards the  study of general SBC dynamics, it has independent relevance. 

The central focus of this work is on characterizing the evolution of the opinion difference for two-agent SBC dynamics over a finite time window. We derive concentration bounds for the opinion difference under sub-Gaussian noise, using a Chernoff bound. In particular, we demonstrate that the opinion difference of the agents is well concentrated around zero, under the same technical conditions that imply the asymptotic stability of the SBC dynamics. 

This paper is organized as follows. In the next section (Sec.~\ref{sec:SBC}), we briefly discuss SBC dynamics. Our main result on the high probability bound on the opinion difference at a finite time when the noise (errors) in opinion exchange has a sub-Gaussian distribution is presented in Sec.~\ref{sec:mainResult}.  In Sec.~\ref{sec:proofOutline}, we provide an outline of the proof of the main result, starting with the relatively simpler case of bounded noise in Sec.~\ref{sec:boundedN}, followed by its extension to sub-Gaussian noise in Sec.~\ref{sec:subG}. Detailed proofs are in Appendix.

% {\color{red}Study of opinion dynamics..dates back to simple mean field model...graph (with the knowledge of social graph) based linear dynamics...then bounded confidence...then what was missing in brief and SBC...what is known, what is not known...what we do here...}

% {\color{red}AC will do: Combine the current intro (after shortening) and a short discussion on SBC opinion dynamics. Say what is known in English: stability, etc. However, nothing is known about the opinion differences at finite time or even its stationary distribution....Say, it is a hard problem...we start with simple two agent cases...which is also not known so far...getting a result for the two agent case  would be an important first step towards solving the general case...}

% {\color{red}Then follow it up with contributions and organizations.}

\section{Stochastic bounded confidence opinion dynamics}
\label{sec:SBC}
Stochastic bounded confidence (SBC) opinion dynamics \cite{Baccelli2017} captures the effect of the social graph as well as that of the closeness of opinions on opinion evolution. Furthermore, it models the impact of inherent stochasticity in human interactions and the unavoidable noise and errors in opinion exchanges. It is a general framework for opinion dynamics that captures the well-known linear dynamics and the bounded confidence dynamics as special cases. Below, we briefly describe a simplified version of the dynamics that is sufficient for the purpose of this work. Please see \cite{Baccelli2017} for a more general description. 

In SBC dynamics, there is an underlying undirected social network $\mathcal{G}=(\mathcal{V}, \mathcal{E})$ of $n$ agents, which captures the possible interactions. Agents $u$ and $v$ can interact only if they share an edge in the social network $\mathcal{G}$, i.e., $(u,v) \in \mathcal{E}$. For each undirected edge $(u,v) \in \mathcal{E}$, there is an influence function $G_{v,u}:[0,\infty) \to [0,1]$, which captures the probability of mutual influence as a function of the opinion difference.

Opinions are real valued and evolve as a discrete time stochastic dynamics. Opinions at time $t$ are denoted by $\{X_u(t): u \in \mathcal{V}\}$. Any two agents $u$ and $v$ with $(u,v) \in \mathcal{E}$ interact at time $t$ with a non-zero probability and at any time, an agent interacts with at most one other agent. If $u$ and $v$ interact at time $t$, then agents are influenced by each other with probability $G_{u,v}(|X_u(t)-X_v(t)|)$. 

If $u$ and $v$ are influenced by each other at time $t$, then they update their opinions as
\begin{align*} 
X_u(t+1) &= \frac{X_u(t)+X_v(t)}{2} + n_u(t), \nonumber \\
X_v(t+1) &= \frac{X_u(t)+X_v(t)}{2} + n_v(t), \nonumber
\end{align*}
and if $u$ or $v$ is not influenced by the other agent, then its opinion evolves as
\begin{align*} 
X_u(t+1) &= X_u(t) + n_u(t). \nonumber
\end{align*}
Here, for each agent $u$, $n_u(t)$ is an i.i.d. zero mean process. This captures the errors and noise in the interactions, which stem from miscommunications and misinterpretations. This also captures the innate  evolution of the opinion of an agent due to his/her own thoughts and emotions. 

Note that if one chooses $G_{u,v}(x)=1$ for all $x$ and $(u,v) \in \mathcal{E}$ and $n_u(t)=0$, we get back the well-known linear dynamics on a social network. On the other hand, if we choose $\mathcal{G}$ to be a clique and for any $u,v$, choose $G_{u,v}(x)=1$ for $x\le d$  and $G_{u,v}(x)=0$ for all $x>d$, we get back the well-known pairwise bounded confidence opinion dynamics. Thus, these two popular class of dynamics are special cases of the SBC dynamics.

As the SBC dynamics is more general, and is nonlinear and stochastic,  analyzing its behavior is significantly more challenging. Due to the presence of noise or error in the SBC dynamics, a consensus cannot be reached, not even in an almost sure or a high probability sense. This in a way reflects the real social dynamics, where there is rarely a consensus. In such a scenario, just like in real democratic societies, we can at best hope for the differences of opinions to remain finite. To capture this scenario, the notion of stability of  opinion dynamics was introduced in \cite{Baccelli2017} and conditions for the stability of SBC dynamics  were established. 

Mathematically, the stability of SBC dynamics is equivalent to the opinion differences between agents reaching a proper stationary distribution. On the other, the dynamics is said to be not stable if the opinion differences do not reach a stationary distribution, which captures the case when opinions of two groups of agents diverge away. The stability results (and their converses) in \cite{Baccelli2017} give conditions in terms of $\mathcal{G}$ and $\{G_{u,v}\}$ for having stable (and not stable) SBC dynamics.

Though the stability results are very important for understanding the highly nonlinear SBC dynamics, they do not shed light on the magnitude of opinion differences in the stable case.
The latter is a quite challenging problem, primarily due to the nonlinear nature of the stochastic dynamics. Hence, as a first step, in this paper, we study a relatively simpler case: SBC dynamics with {\em two agents}. As it will be apparent from the later part of this paper, analyzing the two-agent dynamics requires significant effort. Hopefully, insights obtained from the two-agent case would be useful in analyzing the multi-agent dynamics.

\section{Main Result and Discussion}\label{sec:mainResult}
We consider the following simple two-agent dynamics. Agents $1$ and $2$ interact at all time instants. Thus, at time $t$, if their opinions are $X_1(t)$ and $X_2(t)$, then they are influenced by each other with probability $G(|X_1(t)-X_2(t)|)$. Upon influence, their opinions are updated as, for $i=1,2$,
\[
X_i(t+1) = \frac{X_1(t)+X_2(t)}{2} + n_i(t)\]
and when they are not influenced, the opinions evolve as, for $i=1,2$,
\begin{align*} 
X_i(t+1) &= X_i(t) + n_i(t). \nonumber
\end{align*}

Their opinion difference $Y(t):=X_1(t)-X_2(t)$ is a discrete-time stochastic process and its evolution can be written as: given $Y(t)=y$, $Y(t+1)=\tilde{n}(t)$ with probability $G(|y|)$ and with probability $1-G(|y|)$,
\[Y(t+1) = Y(t) + \tilde{n}(t),\]
where $\tilde{n}(t)=n_1(t)-n_2(t)$ is the difference of two independent zero mean i.i.d. noise processes. We assume that $\tilde{n}(t)$ has a symmetric distribution about its mean.

It was shown in \cite{Baccelli2017} that the opinion difference $Y(t)$ reaches a stationary distribution, i.e., the SBC dynamics is stable, if for some $\delta>0$, $G(x) \gtrsim \frac{1}{x^{2-\delta}}$, where the notation $f(x)\gtrsim g(x)$ means $\liminf_{x\to\infty} \frac{f(x)}{g(x)}>0$. It was also shown that the dynamics is not stable if for some $\delta>0$, $\frac{1}{x^{2+\delta}} \gtrsim G(x)$.  

In this paper, our main result is a bound on the tail probability of a stable two-agent dynamics at a finite time. We establish the bound for the class of sub-Gaussian i.i.d. noise processes.

\begin{definition}[{Sub-Gaussian Random Variable \cite[Sec. 2.3]{lugosi}}]

A random variable $X$ with $\mathbb{E}[X]=0$ is sub-Gaussian with variance parameter $\sigma^2$, denoted by $X \in \mathcal{SG}(\sigma^2)$, if $\; \forall \lambda \in \mathbb{R}$,
\begin{align*}
    \mathbb{E}[\exp{(\lambda X)}] \leq \exp{\Big(\frac{\lambda^2 \sigma^2}{2}\Big)}.
\end{align*} 
% which implies that for every $m>0$,
% \begin{align*}
%     \mathbb{P}(X>m)&\leq \exp{\Big(-\frac{m^2}{2\sigma^2}\Big)}\\\mathbb{P}(X<-m)&\leq \exp{\Big(-\frac{m^2}{2\sigma^2}\Big)}
% \end{align*}
\end{definition}

The main result of this paper can be presented as the following simple statement. 
\begin{proposition}
\label{prop:mainresult}

For a two-agent stochastic bounded confidence dynamics with $G(x) \gtrsim \frac{1}{x^{2-\delta}}$ for some $\delta>0$ and i.i.d.  $\tilde{n}(t)\in\mathcal{SG}(\sigma^2)$ for some finite $\sigma$, \[\mathbb{P}_0(|Y(t)| \ge c~t^{\frac{1}{2}-\beta}) \le 4c_1t^2\exp{(-c_2t^{\frac{\delta}{6}-\frac{2\beta}{3}})}\]
for $\beta<\frac{\delta}{4}$ and positive (independent of $t$) constants $c$, $c_1$ and $c_2$. Here, $\mathbb{P}_0(\cdot)$ corresponds to the conditional probability given that the initial difference $Y(0)=0$.

On the other hand, for a two-agent stochastic bounded confidence dynamics with $G(x) \gtrsim \frac{1}{x^{1-\delta}}$ and $\beta<\frac{\delta}{2}$, %{\color{red}Put the cleaned up version without $\beta'$ (choose a good one or optimize over $\beta'$ to find the best one, as I had told you, and get a lower bound on $c_3$.}
\[\mathbb{P}_0(|Y(t)| \ge c~t^{\frac{1}{2}-\beta}) \le 4c_1t\exp{(-c_2t^{\frac{\delta}{2}-\beta})}.\]
\end{proposition}

Consider the process $S(t)=\sum_{\tau=1}^t n^{(b)}(\tau)$, sum of i.i.d. bounded noise $n^{(b)}(\tau)$. We know that for any $\epsilon>0$ and $c>0$ \cite{lugosi} ,
\[\mathbb{P}(|S(t)| \ge c~t^{\frac{1}{2}-\epsilon})=1-o(1).\]
Thus, the first and obvious implication of Proposition~\ref{prop:mainresult} is that the opinion difference in a stable dynamics concentrates around $0$ and the concentration is much stronger than that of the sum of i.i.d. bounded noise. 

A direct corollary of Proposition~\ref{prop:mainresult} is that, for some $a, b>0$,
\[\mathbb{P}_0(\bigcup_{\tau=t}^\infty \{|Y(\tau)| \ge c~\tau^{\frac{1}{2}-\beta}\}) \le a t^2\exp(-b t^{{\frac{\delta}{6}-\frac{2\beta}{3}}}),\]
which follows using union bound.
This means that the probability of the event that $|Y(\tau)|$ does not cross $\tau^{\frac{1}{2}-\beta}$ after $\tau=t$ approaches $1$ fast as $t\to \infty$. This, in turn, implies that $Y(t)$ almost surely remains within an envelope of the shape $t^{\frac{1}{2}-\beta}$ for $\beta>0$.

% Further, this implies that for $G(x) \gtrsim \frac{1}{x^{\delta}}$, we have that for any $0<c'<\delta$,
% \[\mathbb{P}_0(\cup_{\tau=t}^\infty \{|Y(t)| \ge c~\tau^{c'}) \le a \exp(-b~t^{\epsilon}).\]
% This implies that with high probability the opinion difference remains within an envelope that grows slower than any polynomial.

The opinion difference process $Y(t)$ is a Markov process \cite{Baccelli2017}. A high probability bound of the above kind is not uncommon for well behaved  Markov processes.  In fact, Markov chains with fast decaying stationary distribution would generally result in such bounds. However,  we note that $Y(t)$ is structurally quite different from the Markov chains observed in applications like queuing systems and population dynamics. 

The Markov process $Y(t)$ lies in the class of asymptotically drift zero Markov processes, i.e., its expected drift at $Y(t)=y$ tends to zero as $|y|\to \infty$. This is because the probability of influence decays with increasing opinion difference. Note that even the stable dynamics have asymptotically zero drift. 
Furthermore, unlike queuing processes, $|Y(t)|$ has unbounded jumps. Drift zero Markov processes, as well as Markov processes with unbounded jumps, are uncommon among popular Markov chains seen in the area of information networks.

Though the phenomenon of zero drift is in direct opposition to strong concentration, unbounded jumps towards zero are conducive to concentration. However, the jumps are probabilistic and the probability decays with increasing opinion difference. In that sense, it is not clear whether $Y(t)$ should have a strong concentration around zero or not. Thus, it can be said that Proposition~\ref{prop:mainresult} settles this interesting dilemma affirmatively.

\section{Proof Outline}
\label{sec:proofOutline}
In this section, we provide intuition behind the proof techniques and outline the proof of Proposition~\ref{prop:mainresult}. We start with a relatively simpler case of bounded $\tilde{n}(t)$ and discuss the basics of our proof technique in this context. Then, we extend that proof technique to sub-Gaussian noise in Proposition~\ref{prop:mainresult}.

\subsection{Bounded Noise: A Useful First Step}
\label{sec:boundedN}
As discussed above, the opinion difference process $Y(t)$ is a Markov chain \cite{Baccelli2017}
with unbounded state-dependent jumps and asymptotically zero drift. As a result, its analysis is intricate than the usual Markov chains \cite{Baccelli2017}. So, towards proving the main result in Proposition~\ref{prop:mainresult}, as a first step, we consider a relatively simpler setting: $\tilde{n}(t)$ is bounded, i.e., it has a support $[-D, D]$ for some $D>0$ and $G(x) \gtrsim \frac{1}{|x|^{1-\delta}}$.

%For $G=0$ and $Y(0)=0$, $Y(t)$ is the sum of i.i.d. zero mean noise, i.e., $\mathbb{P}(c_1 \sqrt{t} \le |Y(t)| \le c_2 \sqrt{t}|Y(0)=0)=1-o(1)$ for some $c_1, c_2>0$. However, for $G(x) \gtrsim \frac{1}{|x|^{1-\delta}}$, $Y(t)$ has a much stronger concentration around $0$.

\begin{proposition}
\label{prop:boundedSimpleStatement}
Consider a two-agent stable dynamics with $G(x) \gtrsim \frac{1}{|x|^{1-\delta}}$ for some $\delta>0$, and bounded noise model with the assumed characteristics. Let $k=c~t^{1/2-\beta}$ for some $\beta<\frac{\delta}{2}$ and $c>0$. Then, for $c_1, c_2>0$,
\begin{align*}
    \mathbb{P}_0(|Y(t)| \geq k) \le
    c_1 t \exp{(-c_2 t^{\frac{\delta}{2}-{\beta}})}.
\end{align*}
\end{proposition}

The above result is a direct consequence of the following theorem, and as we discuss later, the derivation of which provides a foundation for the more general case considered in Proposition~\ref{prop:mainresult}. In the rest of this paper, $\mathcal{O}$ represents the big-O notation.

\begin{theorem}\label{theorem_bounded}
Let $k=c~t^{1/2-\beta}$ for some $\beta>0$ and $c>0$. For a two-agent stable dynamics with $G(x) \gtrsim \frac{1}{|x|^{1-\delta}}$ for some $\delta>0$, and bounded noise model with the assumed characteristics, there exists a parameter regime, $\beta<\frac{\delta}{2}$, for which the tail probability of opinion difference decays exponentially with time $t$. That is, for $\lambda=\mathcal{O}(t^{\frac{\delta}{2}-\frac{1}{2}}),\alpha=\mathcal{O}(t^{\delta-1}), \gamma(\lambda)=\mathcal{O}(1)$,
\begin{align*}
    \mathbb{P}_0(|Y(t)|\geq k)) \leq 2\Big(\frac{\gamma(\lambda)}{\alpha}+1\Big) \exp{(-\lambda k)}.
\end{align*}
\end{theorem}

Proof of this result is presented in Appendix~\ref{proof_bound_1}. Let $\mathbb{E}_0[.]$ represent $\mathbb{E}[.|Y(0)=0]$. The main part of the proof involves obtaining a suitable upper-bound on $\mathbb{E}_0[e^{\lambda  Y(t) }]$, which is done in three main steps.

First, for any $t\ge 0$, obtain an upper-bound on $\mathbb{E}_0[e^{\lambda  Y(t+1) }]$ in terms of $\mathbb{E}_0[e^{\lambda  Y(t) }]$, $\mathbb{E}[e^{\lambda  \tilde{n}(t) }]$ and $G(|Y(t)|)$:
\begin{align*}
    \mathbb{E}_0[e^{\lambda  Y(t+1) }]
    &\leq \mathbb{E}[e^{\lambda  \tilde{n}(t) }]\Big(\mathbb{E}_0[e^{\lambda Y(t)}(1-G(|Y(t)|))]+1\Big).
\end{align*}

Second, note that as the noise is bounded, $|Y(\tau)|\le D \tau$ for any $\tau$. Using this fact along with the recursive relation, for any $t>0$, we obtain a bound on $\mathbb{E}_0[e^{\lambda  Y(t) }]$ as sum of products of the moment generating function of noise at $\lambda$, $M_{\tilde{n}}(\lambda)$, and $\{G(D i): 0\le i\le t\}$:
\begin{align*}
    \mathbb{E}_0[e^{\lambda Y(t)}]&\leq M_{\tilde{n}}(\lambda) + \prod_{i=0}^{t-1}M_{\tilde{n}} (\lambda)(1-G(Di))\nonumber\\&+\sum \limits_{i=0}^{t-2} M_{\tilde{n}}(\lambda)^{t-i} \prod_{j=0}^{t-2-i}(1-G(D(t-1-j)).\nonumber
\end{align*}

Third, by using Hoeffding's lemma for the moment generating function of noise, $M_{\tilde{n}}(\lambda)$, and by restricting $\lambda$ to $o(1)$, we obtain the following bound on $\mathbb{E}_0[e^{\lambda  Y(t+1) }]$ in terms of sum of products involving only $\lambda$ and $G(Dt)$:  
\begin{align*}
    \gamma(\lambda)\sum \limits_{i=0}^{t-1} \Big[ \gamma(\lambda)(1-G(Dt))\Big]^{i}\nonumber+\Big[\gamma (\lambda)(1-G(Dt))\Big]^t.\nonumber
\end{align*}

Finally, by the use of the Chernoff bound and a suitable choice of $\lambda$, the bound in Theorem~\ref{theorem_bounded} follows. Please see Appendix \ref{proof_bound_1} for a detailed proof.

The above proof technique has an interesting aspect to it. It starts with an initial weak bound on $|Y(t)|$ and refines it to a much stronger bound. This is also the case in the proof of Proposition~\ref{prop:mainresult}. However, in the latter case, the initial weak bound is tighter and hence, so is the final bound. As a result, the main result applies to all influence functions for which the two-agent SBC dynamics are stable.

\subsection{Extending to Sub-Gaussian Noise}
\label{sec:subG}
Proposition \ref{prop:mainresult} is a consequence of the following more general result. Define an event $A_t=\bigcap\limits_{\tau=h(t)}^t \{|Y(\tau)| \leq d_\tau\}$ where $h(t)$ is order-wise smaller than $t$. Recall that $k=c~t^{1/2-\beta}$ for some $\beta>0$ and $c>0$.
\begin{theorem}\label{theorem_sg}
Consider a two-agent stable dynamics with $G(x) \gtrsim \frac{1}{|x|^{2-\delta}}$ for some $\delta>0$, and sub-Gaussian noise model with the assumed characteristics. Then, with $d_{\tau}=D~\tau^{\frac{1}{2}+\beta'}$ for some $\beta'>0$, there exists a parameter regime, $\beta<\frac{\delta}{4}$, for which the tail probability of opinion difference decays exponentially with time $t$. That is, for $c'>0$, $\gamma(\lambda)=\mathcal{O}(1), \lambda=\mathcal{O}(t^{k'(\frac{\delta}{2}-1)}), \alpha=\mathcal{O}(t^{k'(\delta-2)})$ where $k'=\frac{1}{2}+\beta'$,
\begin{multline*}
    \mathbb{P}_0(|Y(t)| \geq k|Y(0)=0) \leq 2(t-h(t))\exp{(-c'h(t)^{2\beta'})}\\+2\Big[\exp{\Big(\frac{\lambda^2\sigma^2h(t)}{2}}\Big)+\frac{\gamma(\lambda)}{\alpha}\Big]
    \exp{(-\lambda k)}.
\end{multline*}
% Consider a two-agent stable dynamics with $G(x) \gtrsim \frac{1}{|x|^{2-\delta}}$ for some $\delta>0$, and sub-Gaussian noise model with the assumed characteristics. Then, with $d_{\tau}=D\tau^{\frac{1}{2}+\beta'}$ for some $\beta'>0$, there exists $\beta<\frac{\delta}{4}$, $\zeta<1-\frac{\delta}{2}$ for which the tail probability of opinion difference decays exponentially with time $t$. That is, for $c>0$, $\gamma(\lambda)=\mathcal{O}(1), \lambda=\mathcal{O}(t^{k'(\frac{\delta}{2}-1)}), \alpha=\mathcal{O}(t^{k'(\delta-2)})$ where $k'=\frac{1}{2}+\beta'$,
% \begin{multline*}
%     \mathbb{P}_0(|Y(t)| \geq k) \leq 2(t-h(t))\exp{(-ch(t)^{2\beta'})}\\+2\Big[\exp{\Big(\frac{\lambda^2\sigma^2h(t)}{2}}\Big)+\frac{\gamma(\lambda)}{\alpha}\Big]
%     \exp{(-\lambda k)}
% \end{multline*}
% where, $\mathbb{P}_0$ is the probability measure conditioned on $Y(0)=0$.
%  {\color{red}Replace with the better theorem for general $d_{\tau}$.}%.......the upper bound decays for $\beta<\frac{\delta}{2}$ and $\zeta<1-\frac{\delta}{2}$.
\end{theorem}
The first part of proposition \ref{prop:mainresult} follows from Theorem~\ref{theorem_sg} for an appropriate choice of $h(t)=t^\zeta$ with $\zeta<1-\frac{\delta}{2}$ and $\beta'\in \Big(\frac{\frac{\delta}{4}-\beta}{3-\frac{3\delta}{2}},\frac{\frac{\delta}{4}-\beta}{1-\frac{\delta}{2}}\Big)$. The latter part of proposition \ref{prop:mainresult} follows from a similar analysis with $d_\tau=D\tau$.

Proof of Theorem~\ref{theorem_sg} builds on the three step approach discussed in Sec.~\ref{sec:boundedN}. However, proof technique for bounded noise cannot be directly extended to sub-Gaussian noise.
An important ingredient in the proof for the bounded noise case was the bound on $G(|Y(t)|)$ in the second step, which used the fact that $|Y(\tau)|\le D\tau$ for any $\tau$. Clearly, this is not true when the noise is sub-Gaussian. We circumvent this issue by introducing a high probability bound on $|Y(t)|$ instead of a deterministic bound, and by adapting the subsequent proof steps and the final step involving the Chernoff bound according to that high probability bound.

We first discuss the final step involving the Chernoff bound since that would place the changes we make to the three preceding steps in the proof of Theorem~\ref{theorem_bounded} in the right perspective. 

In the proof of Theorem~\ref{theorem_bounded}, the natural Chernoff bound is 
\[\mathbb{P}_0(Y(t)\ge k) \le \mathbb{E}_0[e^{\lambda  Y(t) }] \exp{(-\lambda k)}.\]

In the current setting, we adopt the following useful modification.
\begin{align}
    \mathbb{P}_0(Y(t)\ge k|A_{t-1})  &\le \mathbb{E}_0[e^{\lambda  Y(t) }| A_{t-1}] \exp{(-\lambda k)} \nonumber \\
    \mathbb{P}_0(Y(t)\ge k) & \le \mathbb{P}_0(Y(t)\ge k|A_{t-1})+ 1 - \mathbb{P}_0(A_{t-1}), \label{eq:totalProb}
\end{align}
where, $A_{t-1}$ is an event defined in terms of $\{Y(\tau): 0\le \tau \le t-1\}$. The choice of the event $A_t$ is dictated by the fact that $|Y(t)|$ should be bounded on $A_t$ with high probability and  $1-\mathbb{P}_0(A_{t-1})$ should be rapidly approaching $0$ as $t\to\infty$.

Note that, clearly, $|Y(t)|$ is stochastically smaller than the absolute value of the sum of i.i.d. sub-Gaussian noise $|\sum_{\tau=1}^t \tilde{n}(\tau)|$. Thus, $\mathbb{P}_0(|Y(t)|\ge D~ t^{\frac{1}{2}+\beta'})$
is no more than
\[\mathbb{P}(|\sum_{\tau=1}^t \tilde{n}(\tau)|\ge D~t^{\frac{1}{2}+\beta'}) \le 2~\exp\left(-\frac{D^2}{2 \sigma^2} t^{2\beta'}\right),\]
which fast approaches $0$ as $t\to\infty$. Here, the bound on $\mathbb{P}(|\sum_{\tau=1}^t n(\tau)|\ge D~ t^{\frac{1}{2}+\beta'})$ follows from the concentration inequality for the sum of i.i.d. sub-Gaussian random variables.

Based on the above observations, while keeping the adaptation of the final step involving the Chernoff bound in mind, we adapt the three main steps from the proof of Theorem~\ref{theorem_bounded} as follows.

First, define $A_t=\bigcap\limits_{\tau=h(t)}^t\{|Y(\tau)|\le D~ \tau^{\frac{1}{2}+\beta'}\}$ and obtain an upper-bound on $\mathbb{E}_0[e^{\lambda  Y(t+1) }| A_t]$ in terms of $\mathbb{E}_0[e^{\lambda  Y(t)}|A_t]$, $\mathbb{E}[e^{\lambda  \tilde{n}(t) }]$ and $G(|Y(t)|)$. Then, using the fact that $A_t$ imposes a symmetric constraint on $Y(t)$ and $Y(t)$ has a symmetric distribution, we show that
$\mathbb{E}_0[e^{\lambda  Y(t)}|A_t] \le \mathbb{E}_0[e^{\lambda  Y(t)}|A_{t-1}]$. This gives a recursive upper-bound on  $\mathbb{E}_0[e^{\lambda  Y(t+1) }| A_t]$ in terms of $\mathbb{E}_0[e^{\lambda  Y(t)}|A_{t-1}]$, $\mathbb{E}[e^{\lambda  \tilde{n}(t) }]$ and $G(|Y(t)|)$.

Second, we use the fact that $|Y(t)| \le D~t^{\frac{1}{2}+\beta'}$ on the event $A_t$  to obtain a bound on $\mathbb{E}_0[e^{\lambda  Y(t) }|A_{t-1}]$ as sum of products of the moment generating function of noise at $\lambda$ and $\{G(D~i^{\frac{1}{2}+\beta'}): 0\le i\le t\}$. 

Third, using the sub-Gaussian bound on the moment generating function of noise and by restricting $\lambda$ to $o(1)$, we obtain a bound on $\mathbb{E}_0[e^{\lambda  Y(t) }|A_{t-1}]$ in terms of sum of products involving only $\lambda$ and $G(D~ t^{\frac{1}{2}+\beta'})$. 

Finally, upper-bound on the probability of $A_t$ is obtained using sub-Gaussian concentration and all the bounds are plugged into \eqref{eq:totalProb}.

A detailed proof that essentially formalizes the above steps is presented in Appendix. 

The bound in Proposition~\ref{prop:mainresult} applies to the whole range of influence functions for which dynamics is stable. On the other hand, the bound in Proposition~\ref{prop:boundedSimpleStatement} applies to influence functions satisfying $G(x)\gtrsim \frac{1}{x^{1-\delta}} $. The main reason behind the improvement from Proposition~\ref{prop:boundedSimpleStatement} to Proposition~\ref{prop:mainresult} is the change in the second step of the proof of Theorem~\ref{theorem_bounded}. In the proof for the sub-Gaussian case, the initial bound on $|Y(t)|$ in the second step of the proof of Theorem~\ref{theorem_sg}, though probabilistic, is a significantly tighter bound, which results in a much better final bound. 

\section{Concluding Remarks}
For a stable stochastic bounded confidence opinion dynamics of two agents, we obtained a high probability bound on the opinion difference at a finite time. Our proof technique is based on bounding the conditional moment generating function of the opinion difference on a high probability subset of the sample space and adapting the Chernoff bound accordingly. In future, we look to build on the insights obtained here to address multi-agent dynamics on a social graph.

\section*{Acknowledgement}
The first author's work was supported by the Prime Minister's Research Fellows (PMRF) scheme. The work of AC was partially supported by the Department of Science and Technology, Government of India, under Grant SERB/SRG/2019/001809 and Grant INSPIRE/04/2016/001171.

\bibliographystyle{IEEEtran}
\bibliography{references}

\appendix
\subsection{Proof of Theorem~\ref{theorem_bounded}}
\label{proof_bound_1}
\begin{proof}[\unskip\nopunct]

Assume that $\tilde{n}(t)$ is bounded in $[-D,D]$. Let $M_{\tilde{n}}(\lambda)$ denote its moment generating function at $\lambda$. By using Chernoff bound, for $\lambda>0$,
\begin{align}\label{cher}
    \mathbb{P}_0(\lvert Y(t) \rvert \geq k) = 2\mathbb{P}_0(Y(t) \geq k)  \leq 2\mathbb{E}_0[e^{\lambda  Y(t) }]\exp{(-\lambda k)}.
\end{align}
We now focus on the expectation term in (\ref{cher}). By the law of iterated expectations,
\begin{align*}
   \mathbb{E}_0[e^{\lambda  Y(t+1) }]& = \mathbb{E}_0[\mathbb{E}[e^{\lambda Y(t+1) } | Y(t)]]\\
   &=\mathbb{E}_0[e^{\lambda (Y(t)+\tilde{n}(t))}(1-G(|Y(t)|)+e^{\lambda \tilde{n}(t)}G(|Y(t)|)].
   %&=\mathbb{E}[e^{\lambda \tilde{n}(t)}]\mathbb{E}[e^{\lambda Y(t)}(1-G(|Y(t)|)+G(|Y(t)|)]& \because(\text{for any $t>0$, $\tilde{n}(t)$ and $Y(t)$ are independent})\\
   %&\leq M_{\tilde{n}}(\lambda)\Big(\mathbb{E}[e^{\lambda Y(t)}](1-G(|Y(t)|))+1\Big) &\because(\text{$G(x) \leq 1 \forall x$})%  and in $[0,t]$, $G(|Y(t)|) \geq G(Dt)$})
\end{align*}
For any $t>0$, $\tilde{n}(t)$ and $Y(t)$ are independent. Therefore,
\begin{align*}
    \mathbb{E}_0[e^{\lambda  Y(t+1) }]& = \mathbb{E}[e^{\lambda \tilde{n}(t)}]\mathbb{E}_0[e^{\lambda Y(t)}(1-G(|Y(t)|)+G(|Y(t)|)]\\
    &\leq M_{\tilde{n}}(\lambda)\Big(\mathbb{E}_0[e^{\lambda Y(t)}(1-G(|Y(t)|))]+1\Big).
\end{align*}
As $G(.)$ is decreasing in its argument, in time window $[0,t]$, $G(|Y(t)|) \geq G(Dt)$ for any $t$. We have
\begin{align*}
    \mathbb{E}_0[e^{\lambda  Y(t+1) }]\leq M_{\tilde{n}}(\lambda)\Big(\mathbb{E}_0[e^{\lambda Y(t)}](1-G(Dt))+1\Big).
\end{align*}
This recursive inequality can be expanded to obtain
\begin{align}
    \mathbb{E}_0[e^{\lambda Y(t)}]&\leq M_{\tilde{n}}(\lambda)\nonumber\\&+\sum \limits_{i=0}^{t-2} M_{\tilde{n}}(\lambda)^{t-i} \prod_{j=0}^{t-2-i}(1-G(D(t-1-j))\nonumber\\&+ \prod_{i=0}^{t-1}M_{\tilde{n}} (\lambda)(1-G(Di)).\label{rec_bounded}
    %&\leq  M_{\tilde{n}}(\lambda)+\sum \limits_{i=0}^{t-2} M_{\tilde{n}}(\lambda)^{t-i}(1-G(Dt))^{t-i-1}+\Big[M_{\tilde{n}} (\lambda)(1-G(Dt))\Big]^t\nonumber\\
    %& \leq M_{\tilde{n}}(\lambda)\sum \limits_{i=0}^{t-1} \Big[ M_{\tilde{n}}(\lambda)(1-G(Dt))\Big]^{i}+\Big[M_{\tilde{n}} (\lambda)(1-G(Dt))\Big]^t \label{c_rec}
\end{align}
We exploit the observation that $G(.)$ is decreasing and simplify the RHS of (\ref{rec_bounded}) term by term. Focusing on the second term,
\begin{align*}
    \sum \limits_{i=0}^{t-2} M_{\tilde{n}}(\lambda)^{t-i} &\prod_{j=0}^{t-2-i}(1-G(D(t-1-j))\\&\leq \sum \limits_{i=0}^{t-2} M_{\tilde{n}}(\lambda)^{t-i} \prod_{j=0}^{t-2-i}(1-G(Dt))\\
    &= \sum \limits_{i=0}^{t-2} M_{\tilde{n}}(\lambda)^{t-i}(1-G(Dt))^{t-i-1}\\
    &= M_{\tilde{n}}(\lambda)\sum \limits_{i=1}^{t-1} \Big[ M_{\tilde{n}}(\lambda)(1-G(Dt))\Big]^{i}.
\end{align*}
Now, the third term: 
\begin{align*}
    \prod_{i=0}^{t-1}M_{\tilde{n}} (\lambda)(1-G(Di))
    &\leq M_{\tilde{n}} (\lambda)^t \prod_{i=0}^{t-1}(1-G(Dt))\\
    &= \Big[M_{\tilde{n}} (\lambda)(1-G(Dt))\Big]^t.
\end{align*}
Putting all these terms together in (\ref{rec_bounded}), we get
\begin{align}
    \mathbb{E}_0[e^{\lambda Y(t)}]&\leq M_{\tilde{n}}(\lambda)\sum \limits_{i=0}^{t-1} \Big[ M_{\tilde{n}}(\lambda)(1-G(Dt))\Big]^{i}\nonumber\\&+\Big[M_{\tilde{n}} (\lambda)(1-G(Dt))\Big]^t. \label{c_rec}
\end{align}
By Hoeffding's lemma, 
\begin{align*}
    M_{\tilde{n}}(\lambda) \leq \exp{\Big(\frac{\lambda^2D^2}{2}\Big)}. 
\end{align*}
We assume that $\lambda  \xrightarrow{} 0$ as $t\xrightarrow{} \infty$. That is, for large $t$, $\exp{(\frac{\lambda^2D^2}{2})} = 1 + \frac{\lambda^2D^2}{2} + o(\lambda^4)$ where $\lim \limits_{t\xrightarrow{}\infty} \frac{o(\lambda^4)}{\lambda}=0$. Let $\gamma(\lambda)=1 + \frac{\lambda^2D^2}{2} + o(\lambda^4)$. From (\ref{c_rec}),
\begin{align}
    \mathbb{E}_0[e^{\lambda Y(t)}]&\leq \gamma(\lambda)\sum \limits_{i=0}^{t-1} \Big[ \gamma(\lambda)(1-G(Dt))\Big]^{i}\nonumber\\&+\Big[\gamma (\lambda)(1-G(Dt))\Big]^t.\label{c_bound_e}
\end{align}
Therefore,
\begin{align}
     \mathbb{P}_0(\lvert Y(t) \rvert \geq k) &\leq 2 \Big (\gamma(\lambda)\sum \limits_{i=0}^{t-1} \Big[ \gamma(\lambda)(1-G(Dt))\Big]^{i}\nonumber\\&+\Big[\gamma (\lambda)(1-G(Dt))\Big]^t \Big ) \exp{(-\lambda k)}. \label{c_bound_p}
\end{align}

We observe from (\ref{c_bound_p}) that, as $\lambda$  increases, the exponential term decreases, whereas $\gamma( \lambda)$ increases. One way to choose an optimal $\lambda$ for a better bound is such that $\gamma( \lambda) (1-G(Dt)) <1$.

Based on the stability criterion, we consider $G(x) \gtrsim \frac{1}{|x|^{1-\delta}}$, i.e., for $B, \delta > 0$, $G(|Y(t)|)=\frac{B}{1+|Y(t)|^{1-\delta}}$. For some $\alpha>0$, let $( 1+\frac{\lambda^2D^2}{2})(1-G(Dt))=1-\alpha$ which gives $\lambda = \frac{\sqrt{2}}{D}\sqrt{\frac{B-\alpha(1+(Dt)^{1-\delta})}{(1+(Dt)^{1-\delta}-B)}}$. For $\lambda$ to be real, $\alpha \leq \frac{B}{1+(Dt)^{1-\delta}}$. Pick $\alpha = \frac{B}{2(1+(Dt)^{1-\delta})}$ which in turn gives $\lambda =\frac{1}{D} \sqrt{\frac{B}{(1+(Dt)^{1-\delta}-B)}}$. From (\ref{c_bound_e}), with $\alpha=\mathcal{O}(t^{\delta - 1})$,
\begin{align*}
\mathbb{E}_0[e^{\lambda Y(t) }] &\leq \gamma(\lambda)\frac{1 - (1-\alpha)^t}{\alpha}+1 \leq \frac{\gamma(\lambda)}{\alpha}+1.
\end{align*}
Therefore, 
\begin{align}\label{final_c}
    \mathbb{P}_0(\lvert Y(t) \rvert \geq k) \leq 2\Big(\frac{\gamma(\lambda)}{\alpha}+1\Big) \exp{(-\lambda k)}. 
    %=  \frac{4(1+t^{1-\delta})}{B} e^{-\lambda k}
\end{align}

We have $k=\mathcal{O}(t^{\frac{1}{2}-\beta})$ for some $\beta>0$. 
In (\ref{final_c}), as $\lambda=\mathcal{O}(t^{\frac{\delta}{2} - \frac{1}{2}})$, the bound is useful asymptotically only for $\beta <  \frac{\delta}{2}$, say $\beta =  \frac{\delta}{3}, \frac{\delta}{4}$ and so on. That is, for $\beta <  \frac{\delta}{2}$, (\ref{final_c}) not only verifies with the stability results in the literature but also proves our intuition for $k$. This technique cannot be extended to the case where $G(x) \gtrsim \frac{1}{|x|^{2-\delta}}$ which is quite evident from its available choices of $\beta$.
\end{proof}

\subsection{Proof of theorem \ref{theorem_sg}}
\label{proof_bound_sg}
\begin{proof}[\unskip\nopunct]
We use the following lemmas \ref{conditional_P} and \ref{complement_sg} to prove theorem \ref{theorem_sg}. Proof of the lemmas are detailed in sections \ref{section:conditional_P} and \ref{section:complement_sg} respectively.

Recall the event $A_t=\bigcap\limits_{\tau=h(t)}^t \{|Y(\tau)| \leq d_\tau\}$ where $h(t)=t^\zeta, \; 0<\zeta<1$, and $d_{\tau}=D~\tau^{\frac{1}{2}+\beta'}$ for some $\beta'>0$.
\begin{lemma}\label{conditional_P}
With $G(x) \gtrsim \frac{1}{|x|^{2-\delta}}$ for some $\delta>0$,
\begin{multline*}
    \mathbb{P}_0(|Y(t)| \geq k | A_{t-1})\\ \leq 2\Big[\exp{\Big(\frac{\lambda^2\sigma^2h(t)}{2}}\Big)+\frac{\gamma(\lambda)}{\alpha}\Big] \exp{(-\lambda k)},
\end{multline*}
where $ \gamma(\lambda)=\mathcal{O}(1), \lambda=\mathcal{O}(t^{k'(\frac{\delta}{2}-1)}), \alpha=\mathcal{O}(t^{k'(\delta-2)})$ with $k'=\frac{1}{2}+\beta'$.
\end{lemma}

\begin{lemma} \label{complement_sg}
Recall that $d_{\tau}=D~\tau^{\frac{1}{2}+\beta'}$ for some $\beta'>0$ and $D > 0$. Then, for $c'>0$,
\begin{align*}
    \mathbb{P}_0(A_{t-1}^\mathrm{C})\leq 2(t-h(t))\exp{(-c'h(t)^{2 \beta'})}.
\end{align*}
\end{lemma}
By law of total probability,
\begin{equation*}
    \begin{split}
    \mathbb{P}_0(|Y(t)| \geq k)&=\mathbb{P}_0(|Y(t)| \geq k|A_{t-1})\mathbb{P}_0(A_{t-1})+\\&\quad\; \mathbb{P}_0(|Y(t)| \geq k|A_{t-1}^\mathrm{C})\mathbb{P}_0(A_{t-1}^\mathrm{C})\\&\leq \mathbb{P}_0(|Y(t)| \geq k|A_{t-1})+\mathbb{P}_0(A_{t-1}^\mathrm{C}).
\end{split}
\end{equation*}

From lemma \ref{conditional_P} and lemma \ref{complement_sg},
\begin{align}
    \label{required_bound}
     \mathbb{P}_0(|Y(t)| \geq k)&\leq 2(t-h(t))\exp{(-c'h(t)^{ 2\beta'})}\nonumber\\ &+2\Big(\frac{\gamma(\lambda)}{\alpha}+\exp{\Big(\frac{\lambda^2\sigma^2 h(t)}{2}\Big)}\Big) \exp{(-\lambda k)}.
\end{align}
\end{proof}

\subsection{Proof of Lemma \ref{conditional_P}}\label{section:conditional_P}
\begin{proof}[\unskip\nopunct]
The below mentioned claims \ref{restriction} and \ref{ordering_sg} are useful to prove lemma \ref{conditional_P}.

\begin{claim}\label{restriction}
Recall the event $A_t$. Then, for $\lambda>0$,
\begin{align*}
    \mathbb{E}_0[\exp{(\lambda Y(t))}|A_t] \leq\mathbb{E}_0[\exp{(\lambda Y(t))}|A_{t-1}]. 
\end{align*}
\end{claim}
%\subsection*{Proof of lemma \ref{restriction}}
\begin{proof}[Proof of Claim \ref{restriction}]
Let the conditional probability of $Y(t)$ given an arbitrary event $A$ be denoted as $f_{{Y(t)}|{A}}(.)$. Note that the constraints $\{\lvert Y(\tau) \rvert \leq d_\tau\}$ are symmetric for all $\tau$. With $Y(0)=0$, symmetric noise model and symmetric constraints, we observe that $f_{{Y(t)}|{A_{t}}}(.)$ and $f_{{Y(t)}|{A_{t-1}}}(.)$ are symmetric about zero. Hence,
\begin{align}
    \mathbb{E}_0[\exp{(\lambda Y(t))}|A_{t-1}]&= \int_0^\infty \exp{(\lambda y})f_{{Y(t)}|{A_{t-1}}}(y)dy\nonumber\\&\quad+\int_0^\infty \exp{(-\lambda y})f_{{Y(t)}|{A_{t-1}}}(y)dy\nonumber\\
    &=2 \int_0^\infty \cosh{(\lambda y})f_{{Y(t)}|{A_{t-1}}}(y)dy. \label{At-1}
\end{align}
Similarly, 
\begin{align}
    \mathbb{E}_0[\exp{(\lambda Y(t))}|A_{t}]=2 \int_0^\infty \cosh{(\lambda y})f_{{Y(t)}|{A_{t}}}(y)dy. \label{At}
\end{align}

Note that $\cosh{(\lambda y)}$ increases with $y$ for any $\lambda>0$ and  $f_{{Y(t)}|{A_{t}}}(.)$ is a restriction of $f_{{Y(t)}|{A_{t-1}}}(.)$. Therefore, from (\ref{At-1}) and (\ref{At}), we have the result.
\end{proof}

\begin{claim}\label{ordering_sg}
Let $\{Y'(t),t\ge 0\}$ be the process of opinion difference for a stable two-agent system with $G=0$. That is,
\begin{align*}
    Y'(t+1)=Y'(t)+\tilde{n}(t).
\end{align*}
Then, for any $t \ge 0$,
\begin{align*}
    \mathbb{E}_0[\exp{(\lambda Y(t))]} \leq \mathbb{E}_0[\exp{(\lambda Y'(t))]}\leq \exp{\Big(\frac{\lambda^2\sigma^2 t}{2}\Big)}.
\end{align*}
\end{claim}

%\subsection*{Proof of lemma \ref{ordering_sg}}
The proof of claim \ref{ordering_sg} uses the notion of stochastic ordering and so, we overview it below. For a random variable $X$, let $F_X$ and $\Bar{F}_X$ represent its distribution function and tail distribution respectively, i.e. for any $x\in \mathbb{R}$,
\begin{align*}
    F_X(x)=\mathbb{P}(X \le x),\\
    \bar{F}_X(x)=\mathbb{P}(X > x).
\end{align*}

\begin{definition}[{Stochastic Ordering \cite[Sec. 1.2]{Stoyan1983}}]\label{s_ordering}
Given two random variables $X$ and $Y$ taking values in $\mathbb{R}$, we denote $X \le_{st} Y$ if 
\begin{align*}
    F_X(l) \ge F_Y(l) \; \forall l \in \mathbb{R}
\end{align*}
or equivalently, if
\begin{align*}
    \bar{F}_X(l) \le \bar{F}_Y(l) \; \forall l \in \mathbb{R}.    
\end{align*}
Also, if $X \le_{st} Y$, then $\mathbb{E}[f(X)]\le\mathbb{E}[f(Y)]$ for all non decreasing functions $f$ for which the expectations exist. 
\end{definition}

\begin{proof}[Proof of Claim \ref{ordering_sg}]
Clearly, $Y(t) \leq_{st} Y'(t)$. As exponential function is non decreasing for $\lambda>0$, by definition \ref{s_ordering},  $\mathbb{E}_0[\exp{(\lambda Y(t)) }]\leq \mathbb{E}_0[\exp{(\lambda Y'(t))}]$. Since $Y(0)=0$, 
\begin{align*}
    Y'(t)=\sum_{\tau=0}^{t-1} \tilde{n}(\tau).
\end{align*}
We also assumed $\tilde{n}(t) \in \mathcal{SG}(\sigma^2)$ for all $t$. Therefore, $Y'(t) \in \mathcal{SG}(\sigma^2 t)$. That is,
\begin{align*}
    \mathbb{E}_0[\exp{\lambda Y'(t)}]=\mathbb{E}_0[\exp\Big({\lambda  \sum_{\tau=0}^{t-1} \tilde{n}(\tau)}\Big)]\leq \exp{\Big(\frac{\lambda^2 \sigma^2 t}{2}\Big)}.
\end{align*}
\end{proof}

Along with these claims,  we use the proof technique of theorem \ref{theorem_bounded} to bound the condition probability $\mathbb{P}_0(\lvert Y(t) \rvert \geq k|A_{t-1})$. By the law of iterated expectations,
\begin{align}
   \mathbb{E}_0[e^{\lambda  Y(t+1) }]& = \mathbb{E}_0[\mathbb{E}[e^{\lambda Y(t+1) } | Y(t)]]\nonumber\\
   & =M_{\tilde{n}}(\lambda)\mathbb{E}_0[e^{\lambda Y(t)}(1-G(|Y(t)|))+G(|Y(t)|)]. \label{unconditional_sg}
\end{align}
Now, from (\ref{unconditional_sg}),
\begin{align*}
    \mathbb{E}_0[e^{\lambda  Y(t+1)}|A_t]& =M_{\tilde{n}}(\lambda)\mathbb{E}_0[e^{\lambda Y(t)}(1-G(|Y(t)|))\\ &\quad\quad\quad\quad+G(|Y(t)|)|A_t]\\
    &\leq M_{\tilde{n}}(\lambda)\mathbb{E}_0[e^{\lambda Y(t)}(1-G(|Y(t)|))+1|A_t]\\
    & \labelrel\leq{g_dec} M_{\tilde{n}}(\lambda)(\mathbb{E}_0[e^{\lambda Y(t)}|A_t](1-G(d_t))+1)\\
    & \labelrel\leq{ineq_lemma} M_{\tilde{n}}(\lambda)(\mathbb{E}_0[e^{\lambda Y(t)}|A_{t-1}](1-G(d_t))+1), 
\end{align*}
where inequality \eqref{g_dec} holds as $G(.)$ is decreasing in its argument and inequality \eqref{ineq_lemma} follows claim \ref{restriction}.
So, we have a recursive inequality:
\begin{align*}
\mathbb{E}_0[e^{\lambda  
Y(t+1)}|A_t]\leq M_{\tilde{n}}(\lambda)(1+\mathbb{E}_0[e^{\lambda Y(t)}|A_{t-1}](1-G(d_t))),
\end{align*}
\begin{comment}
Expanding this recursive inequality, we get 
\begin{align}
    \mathbb{E}[e^{\lambda Y(t)}|A_{t-1}]&\leq M_{\tilde{n}}(\lambda)+\sum \limits_{i=h(t)}^{t-2} M_{\tilde{n}}(\lambda)^{t-i} \prod_{j=0}^{t-2-i}(1-G(d_{t-1-j}))+\mathbb{E}[\exp{(\lambda Y(h(t))}] \prod_{i=h(t)}^{t-1}M_{\tilde{n}} (\lambda)(1-G(d_i))\nonumber\\
    &\leq M_{\tilde{n}}(\lambda)+\sum \limits_{i=h(t)}^{t-2} M_{\tilde{n}}(\lambda)^{t-i} \prod_{j=0}^{t-2-i}(1-G(d_{t}))+\mathbb{E}[\exp{(\lambda Y(h(t))}] \prod_{i=h(t)}^{t-1}M_{\tilde{n}} (\lambda)(1-G(d_t))\nonumber\\    
    &\leq  M_{\tilde{n}}(\lambda)+\sum \limits_{i=h(t)}^{t-2} M_{\tilde{n}}(\lambda)^{t-i}(1-G(d_t))^{t-i-1}+\mathbb{E}[\exp{(\lambda Y(h(t))}] \Big[M_{\tilde{n}} (\lambda)(1-G(d_t))\Big]^{t-h(t)}\nonumber\\
    &\leq  M_{\tilde{n}}(\lambda)+M_{\tilde{n}}(\lambda)\sum \limits_{i=h(t)}^{t-2} M_{\tilde{n}}(\lambda)^{t-i-1}(1-G(d_t))^{t-i-1}+\mathbb{E}[\exp{(\lambda Y(h(t))}] \Big[M_{\tilde{n}} (\lambda)(1-G(d_t))\Big]^{t-h(t)}\nonumber\\
    & \leq M_{\tilde{n}}(\lambda)\sum \limits_{i=0}^{t-1-h(t)} \Big[ M_{\tilde{n}}(\lambda)(1-G(d_t))\Big]^{i}+\mathbb{E}[\exp{(\lambda Y(h(t))}] \Big[M_{\tilde{n}} (\lambda)(1-G(d_t))\Big]^{t-h(t)}\nonumber\\
    & \leq M_{\tilde{n}}(\lambda)\sum \limits_{i=0}^{t-1-h(t)} \Big[ M_{\tilde{n}}(\lambda)(1-G(d_t))\Big]^{i}+\exp{\Big(\frac{\lambda^2\sigma^2 h(t)}{2}\Big)} \Big[M_{\tilde{n}} (\lambda)(1-G(d_t))\Big]^{t-h(t)}\ \quad\text{(from lemma \ref{ordering_sg})} \label{c_rec_sg}
\end{align}
\end{comment}
which upon expansion gives
\begin{align}
    \mathbb{E}_0[e^{\lambda Y(t)}|A_{t-1}]\nonumber&\leq M_{\tilde{n}}(\lambda)\\&+\sum \limits_{i=h(t)}^{t-2} M_{\tilde{n}}(\lambda)^{t-i} \prod_{j=0}^{t-2-i}(1-G(d_{t-1-j}))\nonumber\\&+\mathbb{E}_0[\exp{(\lambda Y(h(t))}] \prod_{i=h(t)}^{t-1}M_{\tilde{n}} (\lambda)(1-G(d_i)).\label{rec_sg_main} 
\end{align}
Simplifying the second term in (\ref{rec_sg_main}),
\begin{align*}
    \sum \limits_{i=h(t)}^{t-2} M_{\tilde{n}}(\lambda)^{t-i} &\prod_{j=0}^{t-2-i}(1-G(d_{t-1-j}))\\&\leq
    \sum \limits_{i=h(t)}^{t-2} M_{\tilde{n}}(\lambda)^{t-i} \prod_{j=0}^{t-2-i}(1-G(d_{t}))\\
    &= \sum \limits_{i=h(t)}^{t-2} M_{\tilde{n}}(\lambda)^{t-i}(1-G(d_t))^{t-i-1}\\
    &= M_{\tilde{n}}(\lambda)\sum \limits_{i=h(t)}^{t-2} \Big[M_{\tilde{n}}(\lambda)(1-G(d_t))\Big]^{t-i-1}\\
    &=M_{\tilde{n}}(\lambda)\sum \limits_{i=1}^{t-1-h(t)} \Big[ M_{\tilde{n}}(\lambda)(1-G(d_t))\Big]^{i}.
\end{align*}
Simplifying the product in the third term of (\ref{rec_sg_main}),
\begin{align*}
    \prod_{i=h(t)}^{t-1}M_{\tilde{n}} (\lambda)(1-G(d_i))&\leq \prod_{i=h(t)}^{t-1}M_{\tilde{n}} (\lambda)(1-G(d_t))\\
    &=\Big[M_{\tilde{n}} (\lambda)(1-G(d_t))\Big]^{t-h(t)}.
\end{align*}

Using claim \ref{ordering_sg}, the third term in (\ref{rec_sg_main}) is now upper bounded by $\exp{\Big(\frac{\lambda^2\sigma^2 h(t)}{2}\Big)} \Big[M_{\tilde{n}} (\lambda)(1-G(d_t))\Big]^{t-h(t)}$.
Putting all these simplified terms together in (\ref{rec_sg_main}), we get
\begin{align}
   \mathbb{E}_0[e^{\lambda Y(t)}|A_{t-1}]& \leq M_{\tilde{n}}(\lambda)\sum \limits_{i=0}^{t-1-h(t)} \Big[ M_{\tilde{n}}(\lambda)(1-G(d_t))\Big]^{i}\nonumber\\&+\exp{\Big(\frac{\lambda^2\sigma^2 h(t)}{2}\Big)} \Big[M_{\tilde{n}} (\lambda)(1-G(d_t))\Big]^{t-h(t)}.\label{c_rec_sg}
\end{align}
Since $\tilde{n}(t) \in \mathcal{SG}(\sigma^2)$,
$M_{\tilde{n}}(\lambda) \leq \exp{\Big(\frac{\lambda^2\sigma^2}{2}\Big)}$. We assume that $\lambda  \xrightarrow{} 0$ as $t\xrightarrow{} \infty$. That is, for large $t$, $\exp{(\frac{\lambda^2\sigma^2}{2})} = 1 + \frac{\lambda^2\sigma^2}{2} + o(\lambda^4)$ where $\lim \limits_{t\xrightarrow{}\infty} \frac{o(\lambda^4)}{\lambda}=0$. Let $\gamma(\lambda)=1 + \frac{\lambda^2\sigma^2}{2} + o(\lambda^4)$. Now, (\ref{c_rec_sg}) gives
\begin{align}
    \mathbb{E}_0[&e^{\lambda Y(t)}|A_{t-1}]\leq \gamma(\lambda)\sum \limits_{i=0}^{t-1-h(t)} \Big[ \gamma(\lambda)(1-G(d_t))\Big]^{i}\nonumber\\&+\exp{\Big(\frac{\lambda^2\sigma^2 h(t)}{2}\Big)} \Big[\gamma (\lambda)(1-G(d_t))\Big]^{t-h(t)}.\label{rec_e}
\end{align}
Note that $Y(t)$ is symmetric about zero given $A_{t-1}$. By using Chernoff bound, 
\begin{align}
    \mathbb{P}_0(\lvert Y(t) \rvert \geq k|A_{t-1}) &= 2\mathbb{P}_0(Y(t) \geq k|A_{t-1})\nonumber\\  &\leq 2\mathbb{E}_0[e^{\lambda  Y(t) }|A_{t-1}]\exp{(-\lambda k)}.\nonumber
\end{align}
Therefore,
\begin{align}
    \mathbb{P}_0(\lvert Y(t) \rvert \geq k|A_{t-1}) &\leq 2\Big(\gamma(\lambda)\sum \limits_{i=0}^{t-1-h(t)} \Big[ \gamma(\lambda)(1-G(d_t))\Big]^{i}\nonumber\\+\exp{\Big(\frac{\lambda^2\sigma^2 h(t)}{2}\Big)} &\Big[\gamma (\lambda)(1-G(d_t))\Big]^{t-h(t)}\Big) \exp{(-\lambda k)}. \label{c_bound_p_sg}
\end{align}
We choose an optimal $\lambda$ for a better bound in (\ref{c_bound_p_sg}) by setting $\gamma( \lambda) (1-G(d_t)) <1$. That is, for some $\alpha>0$, let $( 1+\frac{\lambda^2\sigma^2}{2})(1-G(d_t))=1-\alpha$
and with $G(|Y(t)|)=\frac{B}{1+|Y(t)|^{2-\delta}}$, $B, \delta>0$, we get $\lambda = \frac{\sqrt{2}}{\sigma}\sqrt{\frac{B-\alpha(1+d_t^{2-\delta})}{(1+d_t^{2-\delta}-B)}}$. For $\lambda$ to be real, $\alpha \leq \frac{B}{1+d_t^{2-\delta}}$. Choose $\alpha = \frac{B}{2(1+d_t^{2-\delta})}$ and we have $\lambda =\frac{1}{\sigma} \sqrt{\frac{B}{(1+d_t^{2-\delta}-B)}}$. Taking $d_t=Dt^{\frac{1}{2}+\beta'}$ for some $\beta'>0$ gives $\alpha=\mathcal{O}(t^{k'(\delta-2)})$ where $k'=\frac{1}{2}+\beta'$. From (\ref{rec_e}),
\begin{align*}
\mathbb{E}_0[e^{\lambda Y(t) }|A_{t-1}] &\leq \gamma(\lambda)\frac{1 - (1-\alpha)^{t-h(t)}}{\alpha}+\exp{\Big(\frac{\lambda^2\sigma^2 h(t)}{2}\Big)}\\ &\leq \frac{\gamma(\lambda)}{\alpha}+\exp{\Big(\frac{\lambda^2\sigma^2 h(t)}{2}\Big)},
\end{align*}
%\balance
which gives the required tail probability,
\begin{align}\label{final_cond_sg}
    \mathbb{P}_0(\lvert Y(t) \rvert \geq k|A_{t-1}) \leq 2\Big(\frac{\gamma(\lambda)}{\alpha}+\exp{\Big(\frac{\lambda^2\sigma^2 h(t)}{2}\Big)}\Big) \exp{(-\lambda k)}.\nonumber 
    %=  \frac{4(1+t^{1-\delta})}{B} e^{-\lambda k}
\end{align}
From the final inequality,  we observe that, by appropriately choosing $\beta'<\frac{\frac{\delta}{4}-\beta}{1-\frac{\delta}{2}}$, there is a parameter regime, $\beta <  \frac{\delta}{4}$ and $\zeta<1-\frac{\delta}{2}$, such that the tail probability of the opinion difference decays to zero exponentially.
\end{proof}

\subsection{Proof of Lemma \ref{complement_sg}}\label{section:complement_sg}
\begin{proof}[\unskip\nopunct]
We have $A_t=\bigcap\limits_{\tau=h(t)}^t \{\lvert Y(\tau) \rvert \leq d_\tau\}$ where $h(t)=t^\zeta, 0<\zeta<1$. Here, we assume that $d_\tau=D~\tau^{\frac{1}{2}+\beta'}$ for some $\beta'>0$ and $D> 0$.
\begin{align*}
    \mathbb{P}_0(A_{t-1}^\mathrm{C})&=\mathbb{P}_0\Big[\Big(\bigcap\limits_{\tau=h(t)}^{t-1} \{\lvert Y(\tau) \rvert \leq d_\tau\}\Big)^\mathrm{C}\Big]\\
    &\labelrel={demorgan}\mathbb{P}_0\Big[\bigcup\limits_{\tau=h(t)}^{t-1} \{\lvert Y(\tau) \rvert > d_\tau\}\Big]\\
    &\labelrel\leq{unionbound} \sum_{\tau=h(t)}^{t-1}\mathbb{P}_0(\lvert Y(\tau) \rvert > d_\tau)\\
    &\labelrel\leq{ordering}\sum_{\tau=h(t)}^{t-1}\mathbb{P}_0(\lvert Y'(\tau) \rvert > D~\tau^{\frac{1}{2}+\beta'})\\
    &\labelrel\leq{subGofY}\sum_{\tau=h(t)}^{t-1} 2\exp{(-c'\tau^{2\beta'})}\; \text{for some $c'>0$}
\end{align*}
\begin{align*}
    &\leq\sum_{\tau=h(t)}^{t-1} 2\exp{(-c'h(t)^{2\beta'})}.
\end{align*}
Therefore, we have
\begin{align*}
    \mathbb{P}_0(A_{t-1}^\mathrm{C})\leq 2(t-h(t)) \exp{(-c'h(t)^{2\beta'})}.
\end{align*}
\eqref{demorgan} and \eqref{unionbound} follow De Morgan's law and Boole's inequality (union bound) respectively. Recalling the characteristics of $\{Y'(t), t\ge 0\}$ as discussed in the proof of claim \ref{ordering_sg}, we have $\lvert Y(t) \rvert \le_{st} \lvert Y'(t) \rvert$ and $Y'(t) \in \mathcal{SG}(\sigma^2 t)$. Hence, the inequalities \eqref{ordering} and \eqref{subGofY}.
\end{proof}

\end{document}